\documentclass[10pt]{article}
\usepackage[T1]{fontenc}
\usepackage{amsmath}
\usepackage{amsthm}
\usepackage{amssymb}
\usepackage{commath}
\usepackage{multicol}
\usepackage{xparse}
\NewDocumentCommand{\ceil}{s O{} m}{%
  \IfBooleanTF{#1} 
    {\left\lceil#3\right\rceil} 
    {#2\lceil#3#2\rceil} 
}
\usepackage{enumerate}
\usepackage[affil-it]{authblk}

\newcommand{\Z}{{\mathbb Z}}

\newcommand{\C}{{\mathbb C}}

\newcommand{\CC}{{\cal C}}
\newcommand{\NN}{{\cal N}}

\newtheorem{definition}{Definition}
\newtheorem{proposition}{Proposition}
\newtheorem{theorem}{Theorem}

\newtheorem{corollary}{Corollary}

\def\P{|\Psi\rangle}

\DeclareMathOperator{\Tr}{Tr}
\def\floor#1{\left\lfloor #1\right\rfloor}
\makeindex
\title{On the Existence of Absolutely Maximally Entangled States of Minimal Support}

\author{Antonio Bernal%
  \thanks{Electronic address: \texttt{abernal@ub.edu}\\
Supported by project FIS2013-41757-P  }}
\affil{Department de Matemàtiques i Informàtica. Universitat de Barcelona}

\begin{document}
\maketitle

\begin{abstract}
In this paper we prove that no absolutely maximally entangled, AME, state with minimal support exists with 7 sites and 5 levels.

General AME states are pure multipartite states that, when reduced to half or less of the sites, the maximum entropy mixed state is obtained. They have found applications in teleportation and quantum secret sharing, and finding conditions for their existence is a well known open problem. We consider the version of this problem for minimally supported AME states. We single out known both sufficient and necessary conditions in that case. From our negative result, we show that the necessary condition is not sufficient. The proof uses a recent result on the theory of general, nonlinear, classical codes.
\end{abstract}

\begin{multicols}{2}

\section{Introduction}

In this paper we consider pure states of $n$ qudrits, $|\Psi\rangle\in (\C^d)^{\otimes n}$, such that, when tracing out half or more of the sites, the mixed state of maximum confusion is obtained. Those states have been called absolutely maximally entangled, AME, or $AME(n,d)$, in \cite{qss} in the context of quantum secret sharing schemes. The same concept had already appeared in \cite{scott} in the context of quantum error correcting codes, under the term \lq\lq$\floor{n/2}$-uniform''.

AME states have found applications in fields like teleportation or quantum secret sharing, and provide links between different areas of mathematics, like coding theory, orthogonal arrays, quantum error correcting codes or combinatorial designs, see \cite{amecombinatorial}, \cite{qss} and \cite{ameexistenceandapplications}.

A well known open problem is to determine conditions for the existence of AME states. This paper deals with the problem of existence of AME states that are supported on a minimal set of kets from the computational basis.

For AME states of minimal support, a necessary condition is that $d\ge\ceil{n/2}+1$ if $n\ge4$ and $d$ is any integer  \cite{amecombinatorial}, and a sufficient condition is that $d\ge n-1$, when $d$ is a prime power, \cite{ameorthogonal, amecombinatorial, ameexistenceandapplications}.

We prove that there is no $AME(7,5)$ state with minimal support. The result is proved using the standard theory of linear codes, along with a recent result that relates linear and nonlinear codes, see \cite{kokkalaetal}. Since the case where $n=7$ and $d=5$ is not forbidden by the above necessary condition, we see  that the condition is not sufficient.

The organization of the paper is a follows. Sections \ref{amegeneral} and \ref{linearcodes} are devoted to review the general definitions and both necessary and sufficient conditions. In section \ref{negativeexample}, it is proved that no $AME(7,5)$ states of minimal support exist. Section \ref{conclusions} contains concluding remarks and some open questions.

\section{Absolutely maximally entangled states}\label{amegeneral}
Let $n$ and $d$ be integers $n,d\ge 2$. Let $\P$ be a pure multipartite state on $n$ sites, where the local Hilbert space is $d$-dimensional. That is, $\P\in (\C^d)^{\otimes n}$. 
\begin{definition}
We say that $\P$ is absolutely maximally entangled with $n$ sites and local dimension $d$, $AME(n,d)$, if for any partition of $\{1,\ldots,n\}$ into two disjoint subsets $A$ and $B$, with $|B|=m\le|A|=n-m$, the density obtained from $\P\langle\Psi|$ tracing out the sites on the entries in $A$ is multiple of the identity,
\[
\Tr_A\P\langle\Psi| = \frac{1}{d^{m}}Id_{\C^{\otimes m}}.
\]
\end{definition}

If $V$ is a vector space $v\in V$ and ${\cal B}\subset V$ is a basis of $V$, the support of $v$ in the basis $\cal B$ is the number of nonzero coordinates of $v$ in the basis $\cal B$.

A linear algebra argument shows that any $AME(n,d)$ state has support on the computational basis of at least $d^{\floor{n/2}}$.

\begin{definition}
Given two integers $n$, $d$, with $n,d\ge 2$, we will say that an $AME(n,d)$ state $\P$ is of minimal support if the support of $\P$ in the computational basis is $d^{\floor{n/2}}$.
\end{definition}

There is a characterization of $AME(n,d)$ states of minimal support in terms of classical codes.

We consider the set $\Z_d=\{0,\ldots,d-1\}$. A code over the alphabet $\Z_d$ of wordlength $n$ is a subset $\CC\subset\Z_d^n$. On $\CC$ we consider the Hamming distance. Given two words $w,w'\in\CC$, the Hamming distance between $w$ and $w'$, $D_H(w,w')$ is the number of coordinates on which the words $w$ and $w'$ differ. The minimum distance $\delta$ of the code $\CC$ is the minimum of the distances $D_H(w,w')$ between different words $w,w'\in\CC$. The well known Singleton bound establishes that $|\CC|\le d^{n-\delta + 1}$. A code is called maximum distance separable, MDS, if the singleton bound is an equality. See \cite{roth} for general properties of codes.

\begin{theorem}[\cite{amecombinatorial, ameexistenceandapplications}]\label{mdsequivalence}
The existence of $AME(n,d)$ of minimal support is equivalent to the existence of MDS codes of wordlength $n$, alphabet size $d$ and minimum distance $\ceil{n/2}+1$. The words of the code and the kets of the state are in one onto one correspondence.
\end{theorem}

The following property follows by a combinatorial argument involving the associated MDS code.

\begin{proposition}[{\cite{amecombinatorial}}]\label{nimpliesnminusone}
 Let $n\ge 3$ be an integer. If there is an $AME(n,d)$ state of minimal support, then there is an $AME(n-1,d)$ state of minimal support.
\end{proposition} 

So, given $d$, the set of all $n$ such that $AME(n,d)$ states of minimal support exist is an interval.

\begin{corollary}
For any integer $d\ge 2$, there is an integer $\NN(d)$ such that, an $AME(n,d)$ state of minimal support exists if, and only if, $n\le \NN(d)$.
\end{corollary}

We finally mention the necessary condition for the existence of AME states of minimal support:

\begin{theorem}[{\cite{amecombinatorial}}]\label{necessarycondition} 
If $n\ge 4$ and an $AME(n,d)$ state of minimal support exists, then $d\ge\ceil{\frac{n}{2}}+1$.
\end{theorem}

This condition forbids many combinations $(n,d)$ for possible $AME(n,d)$ states of minimal support. For example, although $AME(6,2)$ states exist, none of them can be of minimal support, \cite{amecombinatorial}.

Theorem \ref{necessarycondition} can be read as an upper bound for $\NN(d)$.
\begin{corollary}
For any integer $d\ge 3$, $\NN(d)\le 2d-2$, if $\NN(d)$ is even, and $\NN(d)\le 2d-3$, if $\NN(d)$ is odd.
\end{corollary}

\begin{proof}
We observe that theorem \ref{necessarycondition} is true when $d\ge 3$, for any $n\ge 2$, the cases not covered in theorem \ref{necessarycondition} being trivial. Since $AME(\NN(d),d)$ states of minimal support exist, the statement is another way to write the inequality $\ceil{\NN(d)/2}+1\le d$.\end{proof}

The results discused so far are true for general integer values of the local dimension $d$.

\section{Using linear MDS codes}\label{linearcodes}

In the case where $d$ is a prime power, the alphabet $\{0,\ldots,d-1\}$ can be given a unique field structure, $GF(d)$. In this case, there is more detailed information on certain cases.

In the case of linear $MDS[n,k]$ codes over the field $GF(d)$, where $n$ stands for the code lenght and $k$ is the code dimension, the Singleton identity reads
\[
k=n-\delta + 1,
\]
where $\delta$ is the minimum distance. The linear MDS codes that give rise to $AME(n,d)$ states of minimal support have, according to the Singleton identity and theorem \ref{mdsequivalence}, dimension $k=\floor{n/2}$.

When $d$ is the power of a prime number, we have the theory of generalized Red Solomon, GRS, codes and their extensions, that are known to be MDS. If $4\le n\le d+1$ and $2\le k\le n-2$, there is linear MDS code of lenght $n$ and dimension $k$ over $GF(d)$, see \cite{roth} for details.

The following result gives many examples of $AME$ states of minimal support. It has been stated in \cite{ameexistenceandapplications} resorting to the theory of linear MDS codes, as referred to above, and in \cite{ameorthogonal} using the theory of orthogonal arrays\footnote{Due to a typographical error, the result is stated in \cite{ameorthogonal} for a general integer dimension $d$. The authors ment to state it in the case where $d$ is a prime power.}.
\begin{theorem}[{\cite{ameexistenceandapplications, ameorthogonal}}]\label{someame}
There are $AME(n,d)$ states of minimal support, whenever $n\ge 4$ and $d\ge n-1$ is a power of a prime number.
\end{theorem}

\begin{corollary}\label{corollary}
If $d$ is a prime power, $d\ge 3$, then $\NN(d)\ge d+1$.
\end{corollary}

\section{A negative example}\label{negativeexample}

\begin{theorem}\label{noamecase}
There is no $AME(7,5)$ state of minimal support, $\NN(5)=6$.
\end{theorem}

\begin{proof}

As in \cite{roth}, define $L_d(k)$ as the maximum wordlength of any linear MDS code of dimension $k$ over $GF(d)$, $d$ being a prime power. Several bounds an equalities are known about $L_d(k)$, see \cite{roth}. We will use that $L_d(3)=d+1$ if $d$ is an odd prime power. In particular, we use that $L_5(3)=6$.

This shows that no linear MDS code over $GF(5)$ exists with wordlength 7 and dimension 3.

Now suppose that an $AME(7,5)$ state of minimal support exists. By theorem \ref{mdsequivalence}, there is a MDS code over $GF(5)$ with wordlength 7 and minimum distance $5$.

The code given in theorem \ref{mdsequivalence} however, is not guaranteed to be linear, so this bound $L_5(3)=6$ on the theory of linear codes does not suffice to prove the statement.

To end the proof, we note a result of \cite{kokkalaetal}, that any MDS code, not necessarily linear, over an alphabet of size 5, code size $5^k$, $k\ge 3$, and minimum distance $\delta\ge 3$, can be transformed to a linear MDS code with the same parameters and dimension $k$ with a permutation of coordinates, followed by a permutation of the symbols at each coordinate separately.

This proves that no $AME(7,5)$ state of minimal support exists and $\NN(5)\le 6$, corollary \ref{corollary} gives the reverse inequality.
\end{proof}

The necessary condition given in theorem \ref{necessarycondition} does not forbid the existence of $AME(7,5)$ states of minimal support. This necessary condition, therefore, is no sufficient.

\section{Conclusions}\label{conclusions}

The existence problem for $AME(n,d)$ states is a non trivial one, even for states minimally supported.

$AME(n,d)$ states of minimal support exist if, and only if $n\le\NN(d)$, and the necessary and sufficient conditions reviewed in this paper can be read as:
\[
d+1\le \NN(d) \le 2d-2,\text{ or } 2d-3,
\]
for $d\ge 3$, the inequality on the right being valid for any integer $d$ and the one on the left being valid for all $d$ power of a prime number. We have seen that the upper bound for $\NN(d)$ is not tight, since $\NN(5)=6$.

The theory of linear codes is restricted to the case where the local dimension is a prime power. To investigate other local dimensions, further consideration of general (nonlinear) codes and of combinatorial structures, like orthogonal arrays, seems needed. Sharper estimates on the maximum number of sites $\NN(d)$ for which there are $AME$ states of minimal support for a given local dimension $d$ are desirable too.

\end{multicols}

\end{document}